\documentclass[conference,letterpaper]{IEEEtran}

\IEEEoverridecommandlockouts
\usepackage[utf8]{inputenc} 
\usepackage[T1]{fontenc}
\usepackage{url}
\usepackage{ifthen}
\usepackage{cite}
\usepackage[cmex10]{amsmath} 
\usepackage{cite}
\usepackage{amsmath,amssymb,amsfonts}
\usepackage{algorithmic}
\usepackage{graphicx}
\usepackage{textcomp}
\usepackage{xcolor}
\usepackage{amsthm}
\newtheorem{lemma}{Lemma}

\newtheorem{theorem}{Theorem}

\usepackage{comment}
\usepackage{balance}
\begin{document}

\title{Semantic Freshness Optimal  Sampling and Transmission for Gossiping Receivers}

\author{Irtiza Hasan$\qquad$Ahmed Arafa\\Department of Electrical and Computer Engineering\\ University of North Carolina at Charlotte, NC 28223\\
\emph{ihasan@charlotte.edu}$\qquad$\emph{aarafa@charlotte.edu}
\thanks{This work was supported by the U.S. National Science Foundation under Grant ECCS 21-46099.}}

\maketitle

\begin{abstract}
We study the optimal joint sampling and transmission policy for a transmitter communicating with two gossiping receivers that share information with each other, with the objective of tracking a source under the Version Age of Information (VAoI) metric. The transmitter can observe source-version changes, but it has to pay a sampling cost to get the current source information content. Similarly, it can communicate with a receiver by paying a transmission cost. \emph{Gossiping} enables local information exchange and is able to reduce costly direct transmissions. With imperfect communication links, we formulate an infinite-horizon average-cost Markov Decision Process (MDP) to jointly minimize receiver VAoI, sampling cost, and transmission cost. Using Relative Value Iteration (RVI), we evaluate the optimal policy and establish several properties of its structure. We prove that sampling has a threshold structure in the transmitter VAoI. Among direct transmissions, it is optimal to serve the older receiver. We further characterize the transmit or idle decision through the receiver VAoI difference. Our analysis shows that link reliability and receiver VAoI imbalance have a significant effect on the optimal policy structure. Numerical results verify the structural properties and demonstrate the performance gains of the optimal policy over multiple baselines.
\end{abstract}

%%%%%%%%%%%%%%%%%%%%%%%%%%%%%%%%
% Introduction %
%%%%%%%%%%%%%%%%%%%%%%%%%%%%%%%%

\section{Introduction}
\vspace{-2mm}
Information freshness is considered of critical importance in status-update systems such as sensor networks \cite{ZhengTWC2024}, autonomous vehicles \cite{ChenOJCOMS2025}, and Internet-of-Things (IoT) networks \cite{KahramanJIOT2024} and wireless edge systems \cite{ElshazlyICC2026,AliTCCN2025}, where timely information is needed for control, estimation, and decision-making \cite{SoleymaniCUP2023}. Traditionally, information freshness is quantified using Age of Information (AoI) \cite{AoI}, which has become a central metric for tracking timeliness in recent years \cite{AoISurvey}. However, AoI measures freshness purely through elapsed time and may therefore penalize time staleness even when the underlying source content has not changed and information is fresh. This motivates freshness metrics suitable for semantic communication and semantic freshness \cite{SemCom,UysalSemantic2022}. Version Age of Information (VAoI) is one such semantic freshness metric, introduced in \cite{AoG}, that tracks the destination's lag from the source in terms of generated versions, thus quantifying semantic freshness\cite{SalimnejadTCOMM2025}.

\begin{figure}[t]
    \vspace{-2mm}
    \centering
    \includegraphics[width=0.9\linewidth]{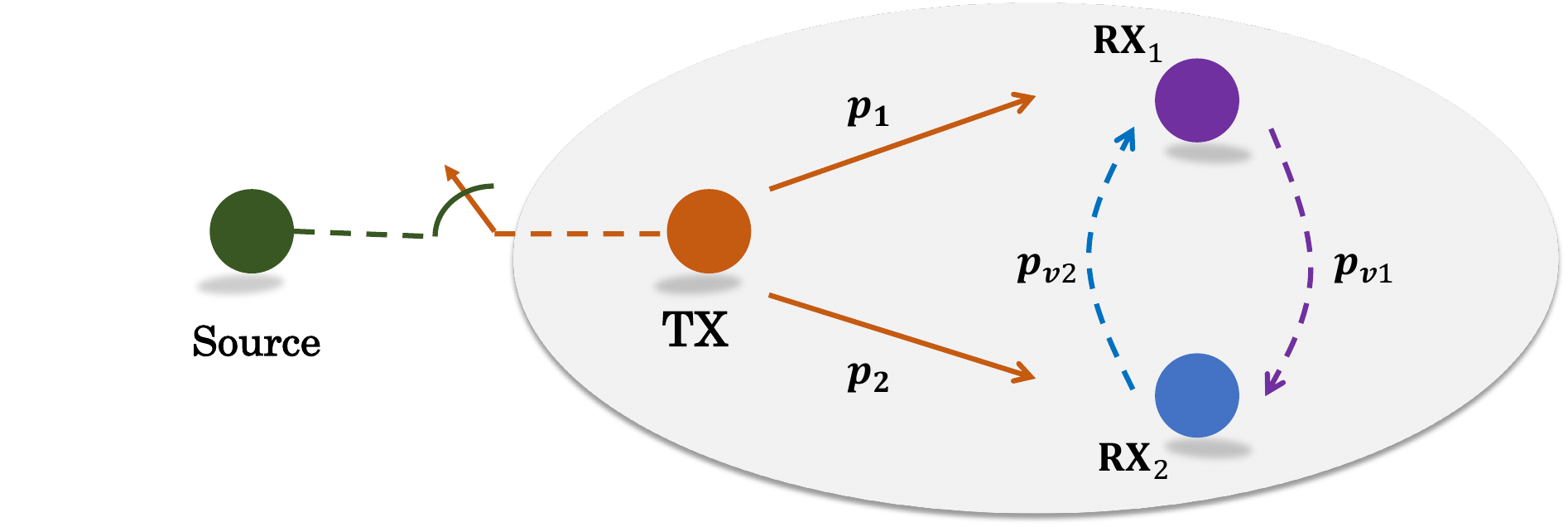} 
    \caption{Transmitter samples the information source (dashed lines and a sampler) and updates two receivers via direct links (solid lines) and receivers may gossip (dashed lines). Symbols above the lines indicate successful communication probabilities.}
    \label{fig:System Model}
    \vspace{-3mm}
\end{figure} 

 \vspace{-0.8mm}
Many networked status-update systems contain local receiver communication links. These networks do not require every update in the network to pass through the central transmitter. We refer to this decentralized information exchange as \emph{gossiping}. Existing works have studied the effects of network topology, cluster formation, mobility patterns, energy constraints, and other relevant properties in gossip networks \cite{VAoIClustered,RCGossip,VAoIContact,PushPullGossip2025,RandomBipartiteGossip2024,MultipleViewsGossip2025,CodedGossip2025}. The structure of optimal freshness policies has been studied extensively in non-gossip network settings \cite{UysalMobiHoc2025,ZakeriTWC2024,SoleymaniTIT2024,AbdElmagidJoint2020,MaatoukAoII2022,ArafaICC2018,ArafaTIT2020,ArafaTWC2019,YangICASSP2020,ArafaTGCN2022,ArafaTCOMM2021,HsuTMC2020,ZhouTCOMM2019}. The advantage of such structural characterizations is that they provide analytical insight which reduces policy-search complexity and paves the way for physics-informed reinforcement-learning methods \cite{NeurIPS2022,BuraTON2026,RoyTGCN2025,ThresholdNW1,ThresholdNW2, StrcuturedRL, JitaniTCCN2022}. Policy optimization for freshness in gossip networks has recently begun to receive attention \cite{MitraICMLCN2024,MitraINFOCOM2023,MitraAllerton2022,VAoIPolicy,TransmitIdleICC2026}.

\vspace{-0.8mm}
In this work, we study a two-receiver VAoI gossip network, for structural characterization of the joint sampling and transmission policy. The transmitter can observe source-version changes, but it has to sample the source at a sampling cost to get the information content. The transmitter can transmit to receivers by paying a transmission cost. \emph{Gossiping} can reduce dependence on costly direct transmissions as receivers spread fresh versions among themselves. The transmitter may idle and allow receiver gossip to reduce VAoI without incurring a direct transmission cost. However, gossip can only circulate information available at the receivers. Thus, the transmitter needs to make a decision whether to transmit or sample/idle to allow gossiping.
\vspace{-0.8mm}

We formulate this as an infinite-horizon average-cost MDP \cite{ZakeriTCOMM2025}. We show that the optimal sampling decision has a threshold structure in the transmitter VAoI and that, among direct transmissions, it is optimal to serve the older receiver. We further characterize the transmit-versus-idle decision through the receiver VAoI difference and show that its threshold direction depends on the relative reliabilities of gossip and direct transmission.

\section{System Model and MDP Formulation}

We consider a time-slotted system indexed by $t\in\{0,1,2,\ldots\}$, consisting of one information source, one transmitter (TX), and two receivers, denoted by RX$_1$ and RX$_2$. Let $X(t)$ denote the source information content in slot $t$. The source remains unchanged from one slot to the next with probability $1-\alpha$ and changes with probability $\alpha$, independently across slots. TX observes source-version changes but does not obtain the new source content unless it samples it. An RX may either receive communication directly from TX, or through gossiping from the other RX. We assume that RXs acknowledge reception of successful communication, regardless of its source, by sending ACK/NACK feedback to TX, thereby allowing it to track the RXs' VAoIs \cite{VAoIPolicy}.

Let $N_s(t)$ denote the source version index at the beginning of slot $t$, and let $N_x(t)$ denote the freshest version index available at node $x$ (TX, RX$_1$ or RX$_2$). The version age of information (VAoI) for node $x$ is defined as
\begin{align}
v_x(t)\triangleq N_s(t)-N_x(t).
\end{align}
Thus, $v_x(t)=0$ means that node $x$ has the source version at slot $t$, while $v_x(t)>0$ means that it lags behind the source by $v_x(t)$ versions, which is its VAoI. We define the system state as
\begin{align}
s(t)=\big(v_{tx}(t),v_1(t),v_2(t)\big),
\end{align}
where $v_{tx}(t)$ is TX's VAoI and $v_i(t)$ is that of RX$_i$, $i\in\{1,2\}$. All RX versions satisfy
\begin{align}
0\le v_{tx}(t)\le v_i(t),\qquad i\in\{1,2\}.
\end{align}
We truncate all VAoIs at $V_{\max}$. 

At each slot, TX selects an action $u$ from the set
\begin{align}
\mathcal A=\{0,1,2,3\},
\end{align}
where $u=0$ denotes idling, $u=1$ and $u=2$ denote transmitting to RX$_1$ and RX$_2$, respectively, and $u=3$ denotes sampling the source. A direct transmission to RX$_i$ succeeds with probability $p_i$. We assume that direct transmission and RX gossip share the same communication resources. Hence, gossip occurs only in slots without a direct transmission. That is, only during idling or sampling slots, RXs gossip with each other, independently. A gossip communication from RX$_1$ (resp. RX$_2$) to RX$_2$ (resp. RX$_1$) succeeds with probability $p_{v1}$ (resp. $p_{v2}$). Let $\delta\in\{0,1\}$ denote the source evolution in a slot, with
\begin{align}
\mathbb P(\delta=1)=\alpha,\qquad
\mathbb P(\delta=0)=1-\alpha.
\end{align}
We define the VAoIs after source evolution as
\begin{align}
\tilde v_{tx} &= (v_{tx}+\delta)\wedge V_{\max},\\
\tilde v_i &= (v_i+\delta)\wedge V_{\max},\qquad i=1,2,\\
\tilde v &= \tilde v_1\wedge \tilde v_2.
\end{align}

Let $\tilde s$ denote the next state. For $u=1$, a successful transmission makes RX$_1$ adopt TX's version, while a failed transmission leaves the receiver versions unchanged:
\begin{equation}
\label{eq:dyn_u1_vai}
\tilde s=
\begin{cases}
(\tilde v_{tx},\,\tilde v_{tx},\,\tilde v_2),
& \text{w.p. } p_1,\\
(\tilde v_{tx},\,\tilde v_1,\,\tilde v_2),
& \text{w.p. } 1-p_1.
\end{cases}
\end{equation}
The transition for $u=2$ follows symmetrically. Under idling, $u=0$, the receivers gossip using their currently available versions. Hence,
\begin{equation}
\label{eq:u0}
\tilde s=
\begin{cases}
(\tilde v_{tx},\,\tilde v,\,\tilde v),
& \text{w.p. } p_{v1}p_{v2},\\
(\tilde v_{tx},\,\tilde v_1,\,\tilde v),
& \text{w.p. } p_{v1}(1-p_{v2}),\\
(\tilde v_{tx},\,\tilde v,\,\tilde v_2),
& \text{w.p. } (1-p_{v1})p_{v2},\\
(\tilde v_{tx},\,\tilde v_1,\,\tilde v_2),
& \text{w.p. } (1-p_{v1})(1-p_{v2}).
\end{cases}
\end{equation}
For $u=3$, TX samples the source and obtains the current source version, resetting its VAoI to zero. In this slot as well, gossiping occurs as in \eqref{eq:u0}. The sampling transition is obtained by replacing the first coordinate $\tilde v_{tx}$ in \eqref{eq:u0} with $0$. 

The one-step cost is the sum receiver VAoI plus the transmitter action cost:
\begin{align}
c(s,u)=v_1+v_2+C(u),
\label{eq:stage_cost_vai}
\end{align}
where
\begin{align}
C(u)=
\begin{cases}
0, & u=0,\\
C_{\mathrm{tx}}, & u\in\{1,2\},\\
C_s, & u=3.
\end{cases}
\label{eq:op_cost_vai}
\end{align}
Here, $C_{\mathrm{tx}}$ and $C_s$ denote the direct transmission and sampling costs, respectively. A stationary policy $\pi:\mathcal S\to\mathcal A$ maps each state to an action. Its long-term average cost starting from $s(0)=s$ is
\begin{align}
\rho^\pi(s)
=
\limsup_{T\to\infty}
\frac{1}{T}
\mathbb E^\pi
\left[
\sum_{t=0}^{T-1}
c(s(t),\pi(s(t)))
\;\middle|\;
s(0)=s
\right].
\label{eq:avg_cost_vai}
\end{align}
The optimal average cost is $\rho^\star=\min_\pi\rho^\pi(s)$ and is independent of the initial state since the finite-state MDP is communicating\cite{Puterman}. We use Relative Value Iteration (RVI) \cite{ZakeriCommLett2025} to obtain $\rho^\star$ and the relative value function $U(s)$ satisfying the Bellman optimality equation \cite{Bertsekas}:
\begin{align}
\rho^\star+U(s)
=
\min_{u\in\mathcal A}
\left\{
c(s,u)
+
\sum_{\bar s\in\mathcal S}
P(\bar s|s,u)U(\bar s)
\right\}.
\label{eq:bellman_vai}
\end{align}
Define
\begin{align}
Q(s,u)
\triangleq
c(s,u)+
\sum_{\bar s\in\mathcal S}
P(\bar s|s,u)U(\bar s).
\end{align}
Then our goal is to characterize an optimal stationary policy
\begin{align}
\pi^\star(s)\in\arg\min_{u\in\mathcal A}Q(s,u).
\label{eq:policy_vai}
\end{align}

%%%%%%%%%%%%%%%%%%%%%%%%%%%%%%%%
% RVI and Structure %
%%%%%%%%%%%%%%%%%%%%%%%%%%%%%%%%
\section{Properties of the Relative Value Function and Structure of the Optimal Policy}

We first establish two monotonicity properties that will be used to
characterize the optimal policy. For states $s=(v_{tx},v_1,v_2)$ and
$s'=(v'_{tx},v'_1,v'_2)$, let $s\preceq s'$ denote componentwise ordering.
For action $u$, let $F_{u,\omega,\delta}(s)$ denote the next state mapping for a realization $\omega$ of the channel outcomes and source
evolution $\delta\in\{0,1\}$. Proofs are omitted due to space constraints.

\begin{lemma}[Order preservation]
\label{lem:F_order_pres}
For any $u\in\mathcal A$, $\delta\in\{0,1\}$, and realization $\omega$ of
the channel outcomes under action $u$, the next state is
order-preserving:
\begin{align}
s\preceq s'
\quad\Longrightarrow\quad
F_{u,\omega,\delta}(s)
\preceq
F_{u,\omega,\delta}(s').
\end{align}
\end{lemma}

\begin{lemma}[Monotonicity]
\label{lem:coord_monotone}
The relative value function $U(s)$ is coordinate-wise nondecreasing:
\begin{align}
s\preceq s'
\quad\Longrightarrow\quad
U(s)\le U(s').
\label{eq:U_mono}
\end{align}
\end{lemma}

We now characterize the optimal sampling decision as a function of the
transmitter VAoI $v_{tx}$.

\begin{theorem}[Sampling threshold]
\label{thm:sampling_threshold}
Fix the receiver VAoIs $(v_1,v_2)$. If sampling is optimal at some feasible
$v_{tx}=k$, then it remains optimal for every feasible $v_{tx}\ge k$.
Therefore, the optimal sampling decision has a threshold structure in
$v_{tx}$.
\end{theorem}

\begin{proof}
Fix $(v_1,v_2)$ and consider the feasible states
\begin{align}
s_k=(k,v_1,v_2),\qquad
k=0,1,\ldots,\min\{v_1,v_2\}.
\end{align}
Thus, $s_i$ and $s_j$ differ only in their transmitter VAoIs $i$ and $j$,
respectively. For the sampling action $u=3$,
\begin{align}
c(s_k,3)=v_1+v_2+C_s,
\end{align}
which is independent of $k$. Moreover, sampling resets the transmitter VAoI
to zero. Hence, for every source realization $\delta\in\{0,1\}$ and every
realization $\omega$ of the gossip outcomes,
\begin{align}
F_{3,\omega,\delta}(s_i)
=
F_{3,\omega,\delta}(s_j),
\qquad \forall i,j.
\end{align}
Moreover, the probability of each realization $(\delta,\omega)$ is the same
from $s_i$ and $s_j$. Thus, every possible next state occurs with the same
probability under the two initial states, and hence
\begin{align}
\mathbb E[U(\bar S)\mid s_i,3]
=
\mathbb E[U(\bar S)\mid s_j,3],
\end{align}
and therefore $Q(s_k,3)$ is independent of $k$. Now consider any non-sampling action $u\in\{0,1,2\}$. For $i\le j$,
\begin{align}
s_i\preceq s_j,
\qquad
c(s_i,u)=c(s_j,u).
\end{align}
By Lemma~\ref{lem:F_order_pres}, for every source and channel realization,
\begin{align}
F_{u,\omega,\delta}(s_i)
\preceq
F_{u,\omega,\delta}(s_j).
\end{align}
Lemma~\ref{lem:coord_monotone} then implies
\begin{align}
Q(s_i,u)\le Q(s_j,u),
\qquad u\in\{0,1,2\}.
\end{align}
Consequently,
\begin{align}
\min_{u\in\{0,1,2\}}Q(s_i,u)
\le
\min_{u\in\{0,1,2\}}Q(s_j,u),
\qquad i\le j.
\end{align}

Define the sampling gap as
\begin{align}
\Delta(k)
\triangleq
Q(s_k,3)
-
\min_{u\in\{0,1,2\}}Q(s_k,u).
\end{align}
Since $Q(s_k,3)$ is constant in $k$, whereas the minimum non-sampling
$Q$-value is nondecreasing in $k$, $\Delta(k)$ is nonincreasing. Hence, if
$\Delta(k)\le0$ at some feasible $k$, then $\Delta(j)\le0$ for every
feasible $j\ge k$, and sampling remains optimal. If sampling is optimal for
at least one feasible value of $k$, the corresponding threshold is
\begin{align}
\tau(v_1,v_2)
\triangleq
\min\{k:\Delta(k)\le0\},
\end{align}
which completes the proof.
\end{proof}

Theorem~\ref{thm:sampling_threshold} separates the sampling decision from the
remaining control problem. The remaining decision is whether to transmit the transmitter's version
to one of the receivers or idle and rely on gossip. To characterize this
decision, we assume symmetric links,
$ 
p_1=p_2=p
$, and $
p_{v1}=p_{v2}=p_v.
$
Under this symmetry, the relative value function is invariant to exchanging
the receiver coordinates.

\begin{lemma}[Symmetry]
\label{lem:symmetry}
The relative value function satisfies
\begin{align}
U(v_{tx},v_1,v_2)
=
U(v_{tx},v_2,v_1),
\quad
\forall (v_{tx},v_1,v_2)\in\mathcal S.
\end{align}
\end{lemma}

The result follows from the fact that swapping the two symmetric receivers with their corresponding transmission actions does not change the stage cost or any success probability.

\begin{lemma}[Older-receiver transmission]
\label{lem:serve_older}
For any feasible state $s=(v_{tx},v_1,v_2)\in\mathcal S$,
\begin{align}
v_1\ge v_2
\quad\Longrightarrow\quad
Q(s,1)\le Q(s,2).
\end{align}
Thus, among the two direct transmissions, it is optimal to serve the
receiver with the larger VAoI.
\end{lemma}

\begin{proof}
The two transmission actions have the same immediate cost and the same
transition upon transmission failure. Upon a successful transmission,
action $u=1$ yields
\begin{align}
(\tilde v_{tx},\tilde v_{tx},\tilde v_2),
\end{align}
whereas action $u=2$ yields
\begin{align}
(\tilde v_{tx},\tilde v_1,\tilde v_{tx}).
\end{align}
By Lemma~\ref{lem:symmetry}, the latter state has the same relative value as
\begin{align}
(\tilde v_{tx},\tilde v_{tx},\tilde v_1).
\end{align}
Since $v_1\ge v_2$ implies $\tilde v_1\ge\tilde v_2$,
Lemma~\ref{lem:coord_monotone} gives
\begin{align}
U(\tilde v_{tx},\tilde v_{tx},\tilde v_2)
\le
U(\tilde v_{tx},\tilde v_{tx},\tilde v_1),
\end{align}
and therefore $Q(s,1)\le Q(s,2)$.
\end{proof}

Hence, in the half-space $v_1\ge v_2$, the only non-sampling comparison
needed is between transmitting to RX$_1$ and idling. Define
\begin{align}
\nabla(k,v_1,v_2)
\triangleq
Q((k,v_1,v_2),1)-Q((k,v_1,v_2),0).
\end{align}
Transmission to RX$_1$ is preferred when
$\nabla(k,v_1,v_2)\le0$, whereas idling is preferred when
$\nabla(k,v_1,v_2)>0$. Let
\begin{align}
\tilde k=(k+\delta)\wedge V_{\max}.
\end{align}
Expanding the two $Q$-values gives
\begin{align}
\nabla(k,v_1,&v_2)
=
C_{\mathrm{tx}}
+
\sum_{\delta\in\{0,1\}}\mathbb P(\delta)
\Big[
p\,U(\tilde k,\tilde k,\tilde v_2)
\notag\\
&
-p_v\,U(\tilde k,\tilde v_2,\tilde v_2)
+(p_v-p)\,U(\tilde k,\tilde v_1,\tilde v_2)
\Big].
\label{eq:tx_idle_gap_vai}
\end{align}

To characterize the effect of receiver imbalance, fix the transmitter VAoI
$k$ and the smaller receiver VAoI $m$, and parameterize the state as
\begin{align}
s_d=(k,m+d,m),\qquad d=v_1-v_2\ge0.
\end{align}
Defining
\begin{align}
g(d)\triangleq\nabla(k,m+d,m),
\end{align}
the expression in \eqref{eq:tx_idle_gap_vai} becomes
\begin{align}
g(d)
&=
C_{\mathrm{tx}}
+
\sum_{\delta\in\{0,1\}}\mathbb P(\delta)
\Big[
p\,U(\tilde k,\tilde k,\tilde m)
-p_v\,U(\tilde k,\tilde m,\tilde m)
\notag\\
&\hspace{2.4cm}
+(p_v-p)\,
U(\tilde k,\widetilde{m+d},\tilde m)
\Big],
\label{eq:g_d_vai}
\end{align}
where $\tilde m=(m+\delta)\wedge V_{\max}$ and $\widetilde{m+d}=(m+d+\delta)\wedge V_{\max}$.

\begin{theorem}[VAoI-difference threshold]
\label{thm:age_difference_threshold}
For fixed transmitter VAoI $k$ and smaller receiver VAoI $m$, $g(d)$ is monotone in the receiver VAoI difference $d$. Consequently, the transmit-versus-idle
decision has a threshold structure in $d$.
\end{theorem}

\begin{proof}
For $d_2>d_1$, subtracting \eqref{eq:g_d_vai} evaluated at the two values
gives
\begin{align}
g(d_2)-g(d_1)
&=
(p_v-p)
\sum_{\delta\in\{0,1\}}\mathbb P(\delta)
\Big[
U(\tilde k,\widetilde{m+d_2},\tilde m)
\notag\\
&\hspace{2.5cm}
-
U(\tilde k,\widetilde{m+d_1},\tilde m)
\Big].
\end{align}
Since $d_2>d_1$ implies
\begin{align}
\widetilde{m+d_2}\ge\widetilde{m+d_1},
\end{align}
Lemma~\ref{lem:coord_monotone} shows that the quantity in brackets is
nonnegative. If $p_v>p$, then $g(d)$ is nondecreasing in $d$. Since transmission is
preferred when $g(d)\le0$, the optimal policy can switch at most once from
transmission to idling as $d$ increases. If $p_v<p$, then $g(d)$ is
nonincreasing, and the policy can switch at most once from idling to
transmission. Finally, if $p_v=p$, the $d$-dependent term in
\eqref{eq:g_d_vai} vanishes, so the transmit-versus-idle decision is
independent of the receiver VAoI difference.
\end{proof}

A useful case occurs when $k=m$, i.e., when the transmitter and the
fresher receiver have the same version. In this case,
$\tilde k=\tilde m$, and \eqref{eq:g_d_vai} reduces to
\begin{align}
g(d)
=
C_{\mathrm{tx}}
+
(p_v-p)
\sum_{\delta\in\{0,1\}}&\mathbb P(\delta)
\Big[
U(\tilde m,\widetilde{m+d},\tilde m)
\notag\\
&
-
U(\tilde m,\tilde m,\tilde m)
\Big].
\label{eq:equal_tx_fresh_rx}
\end{align}
By Lemma~\ref{lem:coord_monotone}, the quantity in brackets is nonnegative.
Therefore, if $p_v\ge p$ and $C_{\mathrm{tx}}>0$, then $g(d)>0$ for every
$d$. Hence, idling is strictly preferred to direct transmission. Thus, when
the transmitter and the fresher receiver hold the same version and gossip is
at least as reliable as the direct link, direct transmission is strictly
suboptimal.

\begin{figure}[t]
    \centering
    \includegraphics[width=0.7\linewidth]{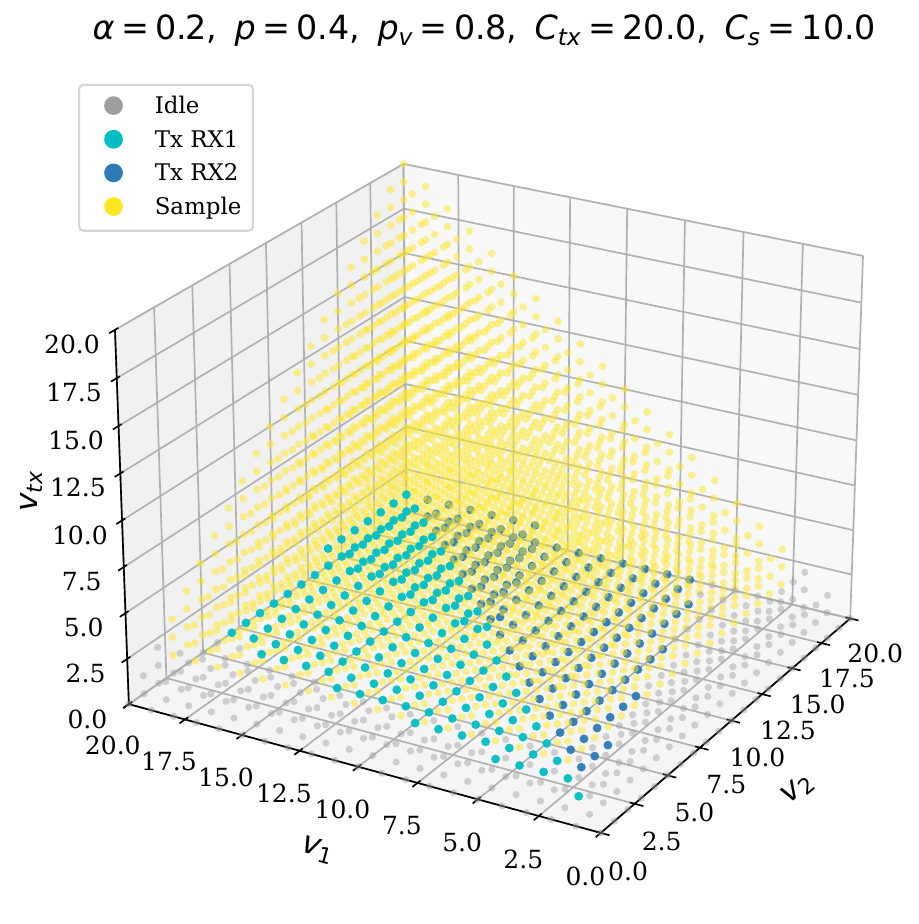}
    \caption{Optimal policy obtained using RVI.}
    \label{fig:3D Policy}
\end{figure}

\begin{figure*}[!t]
    \centering
    \includegraphics[width=0.7\linewidth]{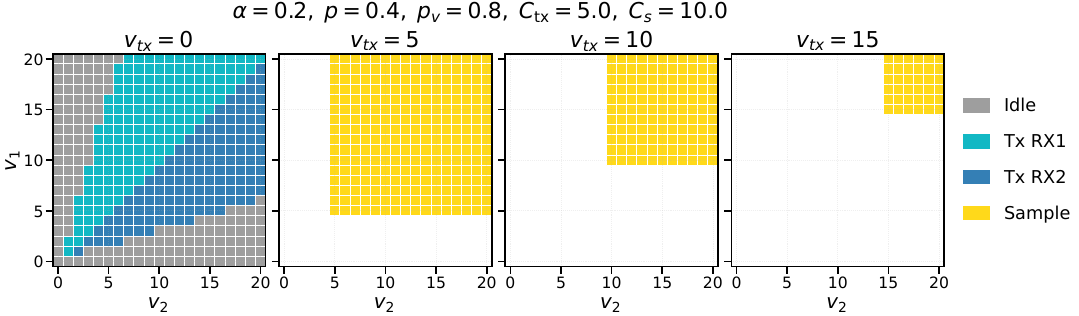}
    \caption{Fixed-$v_{tx}$ two-dimensional slices of the optimal policy in
    Fig.~\ref{fig:3D Policy} over $(v_1,v_2)$ for different values of
    $v_{tx}$.}
    \label{fig:2D Policy}
\end{figure*}

\begin{figure}[t]
    \centering
    \includegraphics[width=0.8\linewidth]{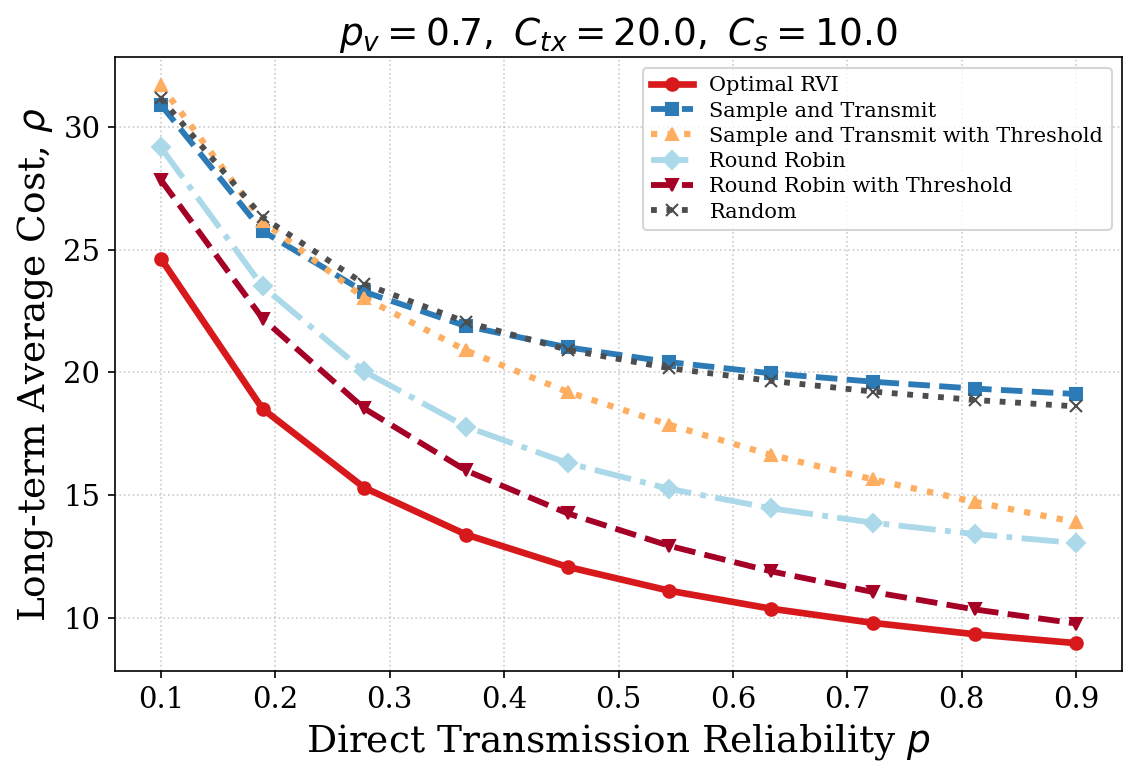}
    \caption{Performance versus direct-transmission reliability.}
    \label{fig:p_sweep}
\end{figure}

\begin{figure}[t]
    \centering
    \includegraphics[width=0.8\linewidth]{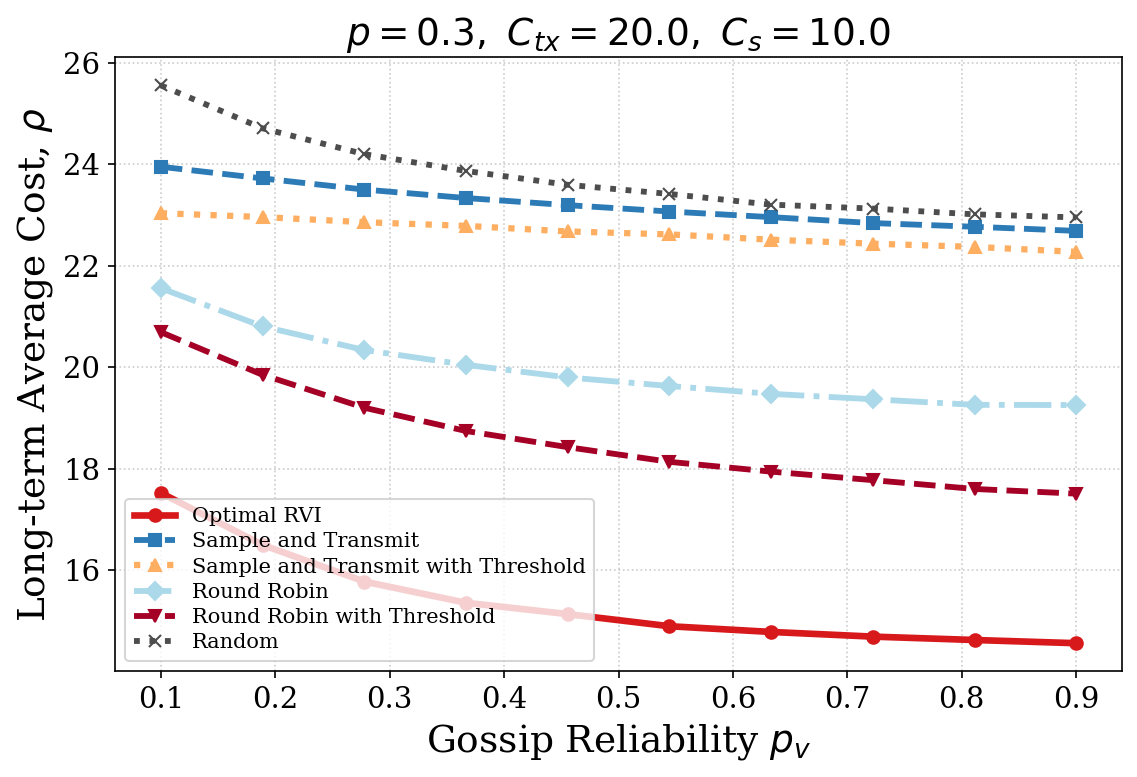}
    \caption{Performance versus gossip reliability.}
    \label{fig:pv_sweep}
\end{figure}

\begin{figure}[t]
    \centering
    \includegraphics[width=0.8\linewidth]{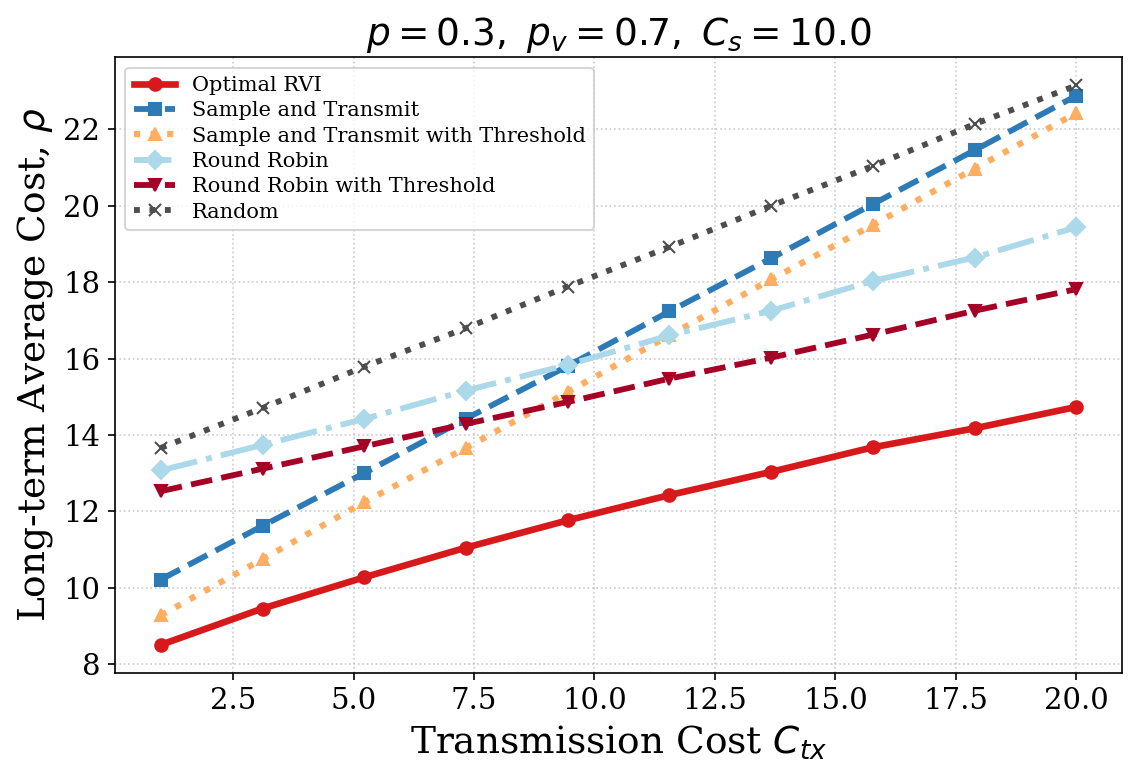}
    \caption{Performance versus transmission cost.}
    \label{fig:Ctx_sweep}
\end{figure}

%%%%%%%%%%%%%%%%%%%%%%%%%%%%%%%%
% Numerical Results %
%%%%%%%%%%%%%%%%%%%%%%%%%%%%%%%%
\section{Numerical Analysis and Policy Visualization}
\label{sec:numerical_results}

We evaluate the optimal joint sampling and transmission policy obtained using
RVI and compare it with five baselines: Sample and Transmit, Sample and
Transmit with Threshold, Round Robin, Round Robin with Threshold, and Random.
Unless varied, the parameters are specified in the corresponding figure. Fig.~\ref{fig:3D Policy} visualizes the optimal policy over
$(v_{tx},v_1,v_2)$, while Fig.~\ref{fig:2D Policy} shows fixed-$v_{tx}$
slices. When $v_{tx}=0$, direct transmission and idling occupy most of the
state space since the transmitter already has the current source version.
The transmission regions are symmetric about $v_1=v_2$, and the older
receiver is selected whenever direct transmission is optimal, consistent
with Lemma~\ref{lem:serve_older}. As $v_{tx}$ increases, the sampling
region expands and eventually dominates the feasible state space, illustrating
the sampling-threshold structure in Theorem~\ref{thm:sampling_threshold}. Fig.~\ref{fig:p_sweep} shows that increasing the direct-link reliability $p$
reduces the long-term average cost, where the optimal policy adapts its transmission decisions to the link quality.
Similarly, Fig.~\ref{fig:pv_sweep} shows that increasing gossip reliability
$p_v$ improves performance by making receiver-to-receiver information dissemination more
effective. These observations are consistent with
Theorem~\ref{thm:age_difference_threshold}, which shows that the relative
reliabilities of gossip and direct transmission determine the structure of
the transmit-versus-idle decision. Fig.~\ref{fig:Ctx_sweep} shows that increasing the transmission cost
$C_{\mathrm{tx}}$ makes direct transmission
less attractive, increasing reliance on gossip. Finally, Fig.~\ref{fig:Cs_sweep} shows that, as $C_s$ increases sampling becomes
less attractive, causing the optimal policy to delay source acquisition until
the freshness benefit justifies the sampling cost. The optimal RVI policy
consistently achieves the lowest long-term average cost among the considered
policies across the parameter sweeps.

\begin{figure}[t]
    \centering
    \includegraphics[width=0.8\linewidth]{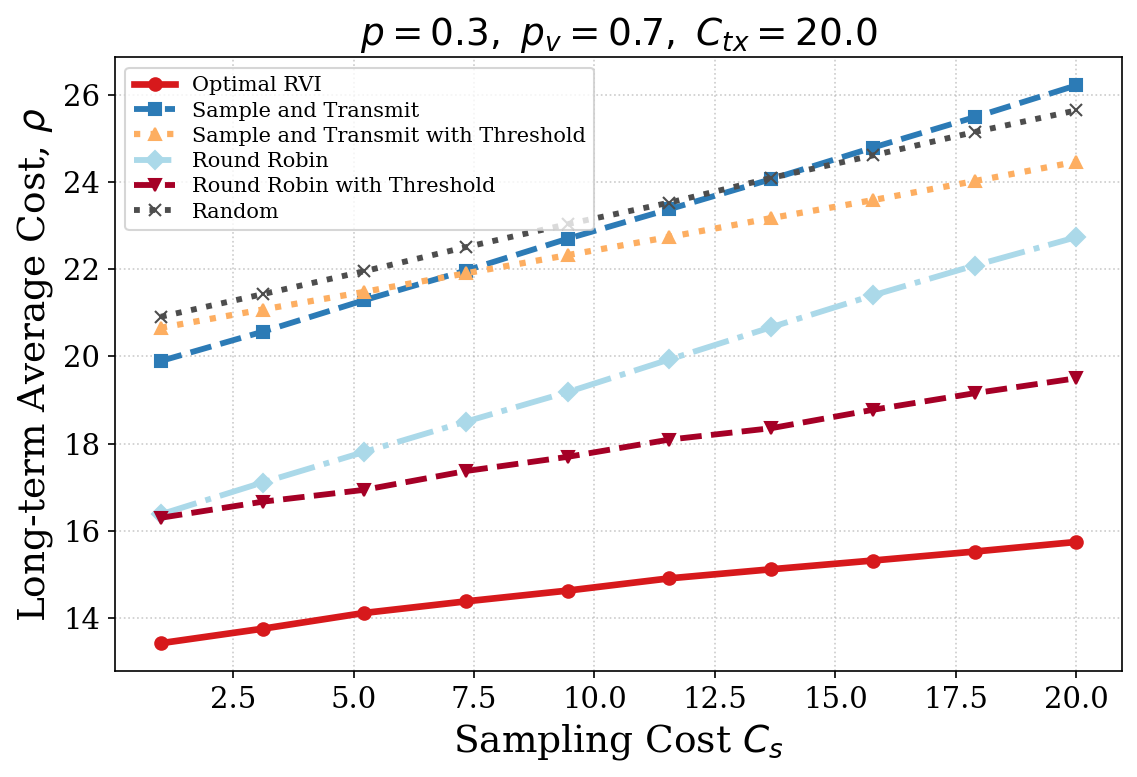} 
    \caption{Performance versus sampling cost.}
     \label{fig:Cs_sweep}
\end{figure}
\section{Conclusion}

We studied joint sampling, transmission, and gossiping for semantic freshness
under VAoI and formulated the problem as an infinite-horizon average-cost
MDP. We showed that the sampling decision has a threshold structure in the
transmitter VAoI and that, among direct transmissions, it is optimal to serve
the older receiver. We further characterized the remaining decision through the receiver VAoI imbalance. In
particular, increasing receiver imbalance can shift the optimal decision from
transmission to idling when $p_v>p$, from idling to transmission when
$p_v<p$, while it has no effect when $p_v=p$.
Finally, when the transmitter and the fresher receiver hold the
same version and gossip is at least as reliable as the direct link, direct
transmission is strictly suboptimal. Numerical results demonstrate the
resulting performance gain of jointly adapting all three decisions. Future work includes extensions
to larger and general gossip networks.

\bibliographystyle{unsrt} 
\bibliography{references}

\end{document}